\documentclass[twocolumn,superscriptaddress,longbibliography,prx,aps]{revtex4-2}

\usepackage{amsfonts,amssymb,amsmath}
\usepackage[]{graphics,graphicx,epsfig}
\usepackage{amsthm}
\usepackage{hyperref}
\usepackage{color}
\usepackage{physics}

\newtheorem{lemma}{Lemma}
\newtheorem{definition}{Definition}
\newtheorem{proposition}{Proposition}

\newcommand{\half}{\mbox{$\textstyle \frac{1}{2}$}}

\newcommand{\proj}[1]{\ket{#1}\bra{#1}}
\newcommand{\identity}{\openone}
\renewcommand{\epsilon}{\varepsilon}

\begin{document}

\title{Quantitative non-classicality of mediated interactions}

\author{Ray Ganardi}
\email{r.ganardi@cent.uw.edu.pl}
\affiliation{Institute of Theoretical Physics and Astrophysics,
Faculty of Mathematics, Physics and Informatics,
University of Gda\'nsk, 80-308 Gda\'nsk, Poland}
\affiliation{Centre for Quantum Optical Technologies, Centre of New Technologies,
University of Warsaw, Banacha 2c, 02-097 Warsaw, Poland}

\author{Ekta Panwar}
\affiliation{Institute of Theoretical Physics and Astrophysics,
Faculty of Mathematics, Physics and Informatics,
University of Gda\'nsk, 80-308 Gda\'nsk, Poland}
\affiliation{International Centre for Theory of Quantum Technologies,
  University of Gda\'nsk,
  ul. Wita Stwosza 63, 80-308 Gda\'nsk, Poland}

\author{Mahasweta Pandit}
\affiliation{Institute of Theoretical Physics and Astrophysics,
Faculty of Mathematics, Physics and Informatics,
University of Gda\'nsk, 80-308 Gda\'nsk, Poland}
\affiliation{Departamento de Física,
Universidad de Murcia, Murcia E-30071, Spain}

\author{Bianka Woloncewicz}
\affiliation{Institute of Theoretical Physics and Astrophysics,
Faculty of Mathematics, Physics and Informatics,
University of Gda\'nsk, 80-308 Gda\'nsk, Poland}
\affiliation{International Centre for Theory of Quantum Technologies,
  University of Gda\'nsk,
  ul. Wita Stwosza 63, 80-308 Gda\'nsk, Poland}
\affiliation{Quantum Research Center, Technology Innovation Institute,
  Masdar City, Abu Dhabi, United Arab Emirates}

\author{Tomasz Paterek}
\affiliation{Institute of Theoretical Physics and Astrophysics,
Faculty of Mathematics, Physics and Informatics,
University of Gda\'nsk, 80-308 Gda\'nsk, Poland}
\affiliation{School of Mathematics and Physics, Xiamen University Malaysia, 43900 Sepang, Malaysia}

\begin{abstract}
In plethora of physical situations one can distinguish a mediator --- a system that couples other, non-interacting systems.
Often the mediator itself is not directly accessible to experimentation, yet it is interesting and sometimes crucial to understand if it admits non-classical properties.
An example of this sort that recently enjoys considerable attention are two quantum masses coupled via gravitational field.
It has been argued that the gain of quantum entanglement between the masses indicates non-classicality of the states of the whole tripartite system.
Here, we focus on non-classical properties of the involved interactions rather than the states.
We derive inequalities whose violation indicates non-commutativity and non-decomposability (open system generalisation of non-commuting unitaries) of interactions through the mediators.
The derivations are based on properties of general quantum formalism and make minimalistic assumptions about the studied systems, 
in particular the interactions can remain uncharacterised throughout the assessment.
Furthermore, we also present conditions that solely use correlations between the coupled systems, excluding the need to measure the mediator.
Next, we show that the amount of violation places a lower bound on suitably defined degree of non-decomposability.
This makes the methods quantitative and at the same time experiment ready.
We give applications of these techniques in two different fields: for detecting non-classicality of gravitational interaction and in bounding the Trotter error in quantum simulations.
\end{abstract}

\maketitle

Mediated interactions are very common and often the mediators are practically inaccessible to direct experimentation.
For example, consider a system of unpaired spins interacting via spin chains in solids~\cite{Sahling2015}.
The bulk measurements of magnetic properties are argued to be solely determined by the unpaired spins at the end of the chain making the chain experimentally inaccessible.
As another example, consider light modes interacting via mechanical membranes~\cite{Thompson2008}.
In this case, usually it is only the light that is being monitored.
Furthermore, fundamentally electric charges are coupled via electromagnetic field, etc.
All these scenarios share a common structure where systems $A$ and $B$ do not interact directly, but are solely coupled via a mediator system $M$, see Fig.~\ref{FIG_SETUP}.
Already at this general level, one can ask about the properties of the mediator that can be deduced from the dynamics of the coupled systems.

Along this line, methods have been proposed to witness non-classicality of the mediator's state from correlation dynamics of the coupled probes.
In particular, conditions were derived under which the gain of quantum entanglement implies that the mediator must have explored non-orthogonal states during the dynamics~\cite{Krisnanda2017,Pal2021}.
Similar ideas applied to more general models than the canonical quantum formalism were used to argue that the entanglement gain between quantum masses witnesses non-classical gravity~\cite{Bose2017,Marletto2017}, and motivated a number of concrete proposals aimed at experimental demonstration of gravity-induced entanglement, see e.g.~\cite{AlBalushi2018,Krisnanda2020,Qvarfort2020,vandeKamp2020,Rijavec2021,Kustura2022,Weiss2021,Carney2021,Pedernales2022,Marshman2022,Christodoulou2023}.
A considerable advantage of these methods is given by minimalistic assumptions they make about the physical systems involved.
They are independent of the initial state, dimensions of involved systems, or the explicit form of interactions and they also work in the presence of local environments.
Accordingly, they are applicable in a variety of fields, see e.g.~\cite{Krisnanda2018} for an example in quantum biology and~\cite{Kon2019} in solid state physics.

Here we move on from the non-classicality of states and develop tools to quantify the amount of non-classicality of mediated \emph{interactions},
while keeping minimalistic assumptions about the considered physical systems.
The notion of non-classicality we employ is given by the commutativity of interaction Hamiltonians, in the case of closed dynamics,
which generalises to decomposability of dynamical maps that also encompasses open systems.
Arguments supporting this choice are given in the next section.
A method to detect presence of such non-classicality was first presented in Ref.~\cite{Krisnanda2018B}, but it was only qualitative, i.e.\ it can only witness the presence of non-commutativity.
It is intriguing that the methods mentioned earlier, aimed at the non-classicality of states, are also at this qualitative level at the present moment.
Our main contribution here is the development of methods to quantify the amount of non-classicality.
We derive conditions which lower bound the norm of the commutator as well as suitably defined distance to decomposable maps.
These conditions are of two types and the structure of the paper reflects this division.
In the first part, we assume that the mediator is accessible to experimentation, and in the second part, the derived conditions use only data measured on the probes.
Non-trivial bounds are derived for any continuous correlation measure.
Hence, it is again expected that the methods presented are applicable in variety of fields.
We provide two examples.

\begin{figure}[!t]
	\centering
        \includegraphics[scale=0.9]{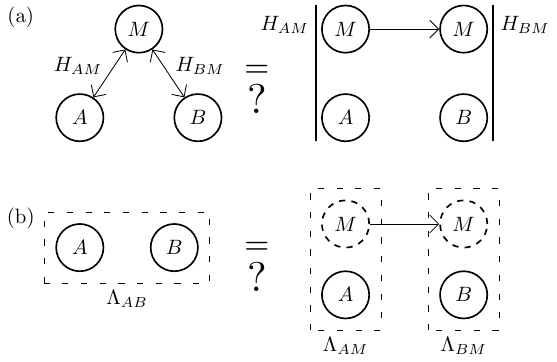}
        \caption{Mediated interactions.
	(a): Systems $A$ and $B$ are coupled via mediator $M$, i.e.\ the underlying Hamiltonian is $H_{AM} + H_{BM}$, and explicitly excludes direct coupling between the systems, i.e.\ $H_{AB}$.
	We present methods based on correlations showing that the interaction Hamiltonians do not commute, i.e.\ the tripartite dynamics cannot be understood
	as a sequence of interactions via $H_{AM}$ and then $H_{BM}$, or in reverse order.
	We also quantify this non-commutativity by providing a lower bound on a suitable norm of the commutator $[H_{AM}, H_{BM}]$.
	These notions are generalised to open systems and we emphasise that the tools make minimalistic assumptions about the whole setup. 
	(b): We extend these techniques to cases where the mediator is non-accessible.
        They are based on correlations in system $AB$ only and show that the tripartite dynamics cannot be understood as a sequence of interactions
        described by dynamical maps $\Lambda_{AM}$ and $\Lambda_{BM}$, or in reverse order.
        We also quantify this form of non-decomposability.
      }
	\label{FIG_SETUP}
\end{figure}

The first one is in the field of quantum simulations.
Suzuki--Trotter expansion is a common way to simulate arbitrary sums of local Hamiltonians, see e.g.~\cite{Lloyd1996,Poulin2015}.
It has been recently shown that the number of Trotter steps needed to obtain required simulation error scales with the spectral norm of the commutator~\cite{Childs2021}.
We link this norm to the correlations in the system, showing a quantitative relation between the complexity of simulation and the amount of correlations.

As the second example, the methods detect and measure non-commutativity of gravitational interaction coupling two quantum masses.
The idea of detecting non-classicality of gravitational interaction has been discussed very recently in Ref.~\cite{Lami2023},
but there the notion of non-classicality is different, based on the impossibility of simulating the dynamics via local operations and classical communication.
Within the quantum formalism, local operations are modelled by arbitrary local channels and classical communication by sequences of dephasing channels connecting the communicating parties. 
In the tripartite setting of two masses and gravitational field, this means the sequence: $\lambda_{AM}$, Dephasing$(M)$, $\lambda_{BM}$, Dephasing$(M)$, etc.
In principle, different dephasing maps could even be performed in different bases. 
In contradistinction, the definition we adopt in the present work deals with continuous in time dynamics and defines classicality at the level of Hamiltonians, as their commutativity. 
This implies an effective picture where a \emph{quantum} mediator is transmitted between `communicating' parties, but only one way. 
So in the tripartite setting, this means $U_{AM} U_{BM}$ or in reversed order.
For other ways of revealing that the evolution cannot be understood in terms of classical (gravitational) field see also Refs.~\cite{Howl2021,Sidajaya2022},
and for general arguments that any system capable of coupling to a quantum system must itself be quantised, see e.g.~\cite{Marletto2022}.
Our tools show that correlations between the masses exclude gravity as interaction with commuting particle-field couplings.


\section{Classicality and decomposability}

Let us start with closed systems and explain our choice of the notion of classicality and its relation to the properties of dynamical maps.
In this work, classical mediated interactions are defined by commuting Hamiltonians $H_{AM}$ and $H_{BM}$, see Fig.~\ref{FIG_SETUP}.
A high level motivation for this choice comes from the fact that in classical mechanics, all observables commute,
hence a classical mediator would have all its couplings to other systems commuting.
The commutativity can also be motivated starting with the notion of classical states as those admitting vanishing quantum discord~\cite{Modi2012}, 
or vanishing coherence in the case of a single system~\cite{Streltsov2017}, and asking for the evolution that preserves this form of classicality.
The vanishing discord means that the whole tripartite state can be measured on the mediator without disturbing the total state.
Mathematically, the state has a block diagonal form and we assume that at all times there exists a single `preferred' basis of the mediator.
We show in Appendix~\ref{APP_NOTION} that such dynamics is generated if and only if the Hamiltonian has a block diagonal form too, with the same basis on the mediator.
Since we consider here systems with global Hamiltonian $H = H_{AM} + H_{BM}$, the state classicality is preserved when both $H_{AM}$ and $H_{BM}$ are block diagonal with the same basis on system $M$, i.e.\ both Hamiltonians commute $[H_{AM}, H_{BM}] = 0$.
Furthermore, for commuting non-degenerate $H_{AM}$ and $H_{BM}$, the total Hamiltonian admits only product eigenstates and out-of-time-ordered correlators vanish at all times, as shown in Appendix~\ref{APP_NOTION}.

A closely related notion is that of decomposability.
A tripartite unitary $U$ is decomposable if there exist unitaries $U_{AM}$ and $U_{BM}$ such that
\begin{equation}
U = U_{BM} U_{AM}.
\end{equation}
Intuitively, decomposable unitaries are those that can be simulated by first coupling one of the systems
to the mediator $M$, and then coupling the other.
One can think the mediator particle is being transmitted between $A$ and $B$ which are in separate laboratories making this setting similar to that in Refs.~\cite{Cubitt2003,Streltsov2012,Chuan2012,Fedrizzi2013,Vollmer2013,Peuntinger2013}.
Although the Suzuki--Trotter formula shows that any unitary can be approximated by a sequence of Trotter steps,
decomposable unitaries are special because we can implement the exact unitary with only a single Trotter step.
For its relation to the notion of locality in quantum field theory, see Ref.~\cite{arxiv_Biagio_2023}.

Clearly, for classical interactions $[H_{AM}, H_{BM}] = 0$, the unitary operator $U(t) = e^{-itH}$ is decomposable for all $t$.
But there exist unitaries that are decomposable and yet they are not generated by a classical interaction.
A concrete example is given in Appendix~\ref{APP_ONEWAY} and relies on the fact that the unitary can be written as $U = U_{BM} U_{AM}$,
but there exist no unitaries $V_{AM}$ and $V_{BM}$ which in sequence $V_{AM} V_{BM}$ would be equal to $U$.
This example already suggests that decomposability has to be augmented with commutativity of decompositions to be equivalent to the classicality of interactions,
a fact that we prove in Appendix~\ref{APP_CL_COMM}.
Therefore, the unitary generated by classical interactions is continuously decomposable, with the added property that the decomposition must commute, i.e.\ $[ U_{AM} (t), U_{BM} (t) ] = 0$ for all $t$.
Accordingly, it is irrelevant whether we define the decomposition order as $U_{BM} U_{AM}$ or $U_{AM} U_{BM}$.

Decomposability naturally extends to open systems.
In this case, the evolution is described by a map $\lambda$ giving the state of the system at time $t$, i.e.\ $\rho = \lambda(\rho_0)$.
We say that a tripartite map $\lambda$ is decomposable if there exist maps $\lambda_{AM}$ and $\lambda_{BM}$ such that
\begin{equation}
\lambda(\rho) = \lambda_{BM} \lambda_{AM} (\rho),
\label{EQ_DEC_MAPS}
\end{equation}
for every $\rho$.
In Appendix~\ref{APP_CONSISTENCY},  we discuss consistency of this definition and the one based on unitaries.
As expected, a unitary operator is decomposable if and only if the corresponding unitary map is decomposable (general maps are not required).

It is this general notion of decomposability that we will exclude and measure the degree of its exclusion in the coming sections.
A number of similar concepts has been introduced before and it is instructive to compare the decomposability with them and note where the novelty is.
So-called divisibility asks whether map $\Lambda$ can be written as $\Lambda_1 \Lambda_2$ where both $\Lambda_1$ and $\Lambda_2$ are not unitaries~\cite{Wolf2008}.
A stronger notion of cp-divisibility, studied in the context of Markovian dynamics~\cite{Lindblad1976,Gorini1976},
asks whether map $\Lambda_t$ can be written as the sequence of completely positive maps $\Lambda_t = V_{t,s} \Lambda_s$.
Interestingly, the set of cp-divisible maps is not convex~\cite{Wolf2008}.
The decomposability we study here has a specific multipartite structure that was considered only in~\cite{Streltsov2015,Krisnanda2018B}, that is clearly significant from physics perspective.


\section{Accessible mediator}

We first present methods that utilise correlations measured on all three subsystems, and devote the next section to eliminating measurements on the mediator.
The basic idea is that correlations between subsystem $A$ and subsystems $MB$ together should be bounded in the case of decomposable dynamics
because they are effectively established via a process where mediator is being transmitted from $A$ to $B$ only \emph{once}.
It is therefore expected that the correlations are bounded by the `correlation capacity' of the mediator, i.e.\ maximal correlations to the mediator alone.
Such inequalities for distance-based correlation measures have been derived in Ref.~\cite{Krisnanda2018B} and could also be obtained by manipulating the results of Refs.~\cite{Streltsov2015,Krisnanda2017}.
Our contribution in this section is a generalisation to any continuous correlation measure and then quantification of non-decomposability based on the amount of violation of the derived criterion.


\subsection{Detecting non-decomposability}

Let us take a correlation quantifier $Q$ that is monotonic under local operations.
In Appendix~\ref{APP_ACCESSIBLE} we show that the bound in terms of correlation capacity holds when we additionally assume the initial state is of the form $\rho_0 = \rho_{AM} \otimes \rho_B$.
For such an initial state, the correlations generated by a decomposable map $\lambda$ admit
\begin{equation}
Q_{A:MB}(\lambda(\rho_0)) \le \sup_{\sigma_{AM}} Q_{A:M}(\sigma_{AM}),
\label{EQ_CAPACITY}
\end{equation}
where the bound is derived for any correlation measure $Q$ that is monotonic under local processing.
Here $\sigma_{AM}$ ranges over all possible joint states of $AM$.
If the initial state is fully product $\rho_0 = \rho_A \otimes \rho_M \otimes \rho_B$, the roles of $A$ and $B$ could be exchanged giving rise to another inequality and the experimenter should choose the one that is violated to detect non-decomposability.
This bound is already non-trivial as we now demonstrate by showing that the maximally entangling map cannot be decomposable.
Consider the initial product state $\ket{000}$ and assume systems $A$ and $B$ are of higher dimension than the mediator, i.e.\ $d_A = d_B > d_M$.
As an exemplary entanglement measure, take the relative entropy of entanglement, $E$.
It is known that its maximum depends on the dimension of the smaller Hilbert space, i.e.\ $\sup_{\sigma_{AM}} E_{A:M}(\sigma_{AM}) = \log d_M$.
According to Eq.~(\ref{EQ_CAPACITY}), any decomposable evolution cannot produce more entanglement than $\log d_M$.
This holds for entanglement $E_{A:MB}$ as well as for $E_{A:B}$ due to the monotonicity of relative entropy under partial trace.
Since dimensions of $A$ and $B$ are larger than the dimension of the mediator, maximally entangled state between $AB$ cannot be produced by any decomposable map.

Of course we are interested in extending Eq.~(\ref{EQ_CAPACITY}) to arbitrary initial state and in this way make the method independent of it.
To achieve this aim, we use continuity arguments.
Many correlation measures, including relative entropy based quantifiers~\cite{Donald1999}, all distance-based measures~\cite{Modi2010} or convex roof extensions of asymptotically continuous functions~\cite{SynakRadtke2006},
admit a version of continuity where there exists an invertible, monotonically non-decreasing function $g$, such that $|Q(x) - Q(y)| \le g(d(x,y))$,
where $d$ is a contractive distance and $\lim_{s \to 0} g(s) = 0$.
This is a refinement of the notion of uniform continuity, where we can bound how much the function varies when we perturb the input.
A notable example is logarithmic negativity~\cite{Plenio_2005} that is not asymptotically continuous, yet fulfills this notion of continuity.
For simplicity, we shall call such functions gd-continuous.
We prove in Appendix~\ref{APP_ACCESSIBLE} that correlation quantifiers which are gd-continuous are bounded in decomposable dynamics as follows:
\begin{equation}
Q_{A:MB}(\lambda(\rho_0)) \le \sup_{\sigma_{AM}} Q_{A:M}(\sigma_{AM}) + I_{AM:B}(\rho_0),
\label{EQ_ABM}
\end{equation}
where $I_{AM:B}(\rho) = \inf_{\sigma_{AM}\otimes \sigma_B} g(d(\rho, \sigma_{AM}\otimes \sigma_B))$ is a measure of total correlations in the state $\rho$ across the partition $AM:B$.
Indeed, from the properties of $g$ and $d$, it is easy to verify that this quantity is monotonic under local operations and it is zero if and only if $\rho$ is a product state across $AM:B$ partition.
Again, an independent inequality is obtained by exchanging $A$ and $B$.

This bound is also non-trivial and its violation has been demonstrated in Ref.~\cite{Ganardi2022}, which focused on negativity as a concrete correlation (entanglement) measure.
The system under consideration involved two cavity modes $A$ and $B$ coupled via two-level atom $M$.
This scenario is particularly well-suited to demonstrate the violation because the dimension of the mediator is as small as it can be whereas the dimensions of the probes are in principle unbounded.


\subsection{Measuring non-decomposability}

Having established witnesses of non-decomposability, we now argue that the amount of violation of Eq.~(\ref{EQ_ABM}) quantifies the non-decomposability.
As a measure of non-decomposability we propose a minimal operator distance from an arbitrary map $\Lambda$ to the set of decomposable maps, that we denote as $\texttt{DEC}$:
\begin{equation}
\mathrm{ND}(\Lambda) = \inf_{\lambda \in \texttt{DEC}} D (\Lambda, \lambda).
\end{equation}
We shall refer to this quantity as the `degree of non-decomposability'.
The operator distance $D$ in its definition could be chosen as the one induced by the distance on states
\begin{equation}
D (\Lambda_1, \Lambda_2) = \sup_\sigma d(\Lambda_1(\sigma), \Lambda_2(\sigma)),
\end{equation}
where $\Lambda_1$ and $\Lambda_2$ are arbitrary maps and $\sigma$ is any state from the domain of the map.
In Appendix~\ref{APP_ACCESSIBLE} we demonstrate that violation of Eq.~(\ref{EQ_ABM}) lower bounds the degree of non-decomposability as follows
\begin{eqnarray}
  \mathrm{ND}(\Lambda) &\ge& g^{-1} ( Q_{A:MB}(\Lambda(\rho_0)) - B(\rho_0) ),
\label{EQ_VIOL_AMB}
\end{eqnarray}
where $B(\rho_0)$ is the right-hand side of Eq.~(\ref{EQ_ABM}).
Accordingly, any violation of the decomposability criterion in terms of correlations sets a non-trivial lower bound on the distance between the dynamical map and the set of decomposable maps.


\subsection{Quantum simulations}
\label{SEC_TROTTER}

As the first application of the introduced measure, suppose we would like to simulate the dynamics generated by the Hamiltonian $H = H_{AM} + H_{BM}$. 
(In fact this analysis can be generalized to any 2-local Hamiltonian).
Quantum simulators implement dynamics close to the desired one by truncating the Suzuki--Trotter formula to $r$ Trotter steps
\begin{equation}
e^{-itH} \approx \left( e^{-i \frac{t}{r} H_{AM}} e^{-i \frac{t}{r} H_{BM}} \right)^r.
\end{equation}
The error of this approximation can be quantified by the spectral norm (the largest singular value)
\begin{equation}
\norm{e^{-itH} - \left( e^{-i \frac{t}{r} H_{AM}} e^{-i \frac{t}{r} H_{BM}} \right)^r}_{\infty}
\end{equation}
and it was shown in Ref.~\cite{Childs2021} that in order to make this error smaller than $\epsilon$, the number of Trotter steps has to scale with the norm of the commutator
\begin{equation}
r = O\left( \frac{t^2}{\epsilon} || \, [H_{AM}, H_{BM}] \, ||_{\infty} \right).
\label{EQ_R}
\end{equation}

Our aim is to provide a lower bound on the commutator norm in terms of correlations, and in this way bound the number of required Trotter steps.
Recall after Ref.~\cite{Childs2021} that for a single Trotter step we have
\begin{eqnarray*}
|| U - U_{AM} U_{BM}||_{\infty}
\le \frac{t^2}{2}  || \, [H_{AM}, H_{BM}] \, ||_{\infty},
\end{eqnarray*}
where $U = e^{-itH}$ and, e.g., $U_{AM} = e^{-i t H_{AM}}$.
We need to link our methods to the spectral norm.
For finite-dimensional systems, all metrics generate the same topology~\cite{Rudin1991}, i.e.\ for any two distances $d_1$ and $d_2$ there exists a constant $C$ such that
\begin{eqnarray}
\frac{1}{C} \, d_2(\rho,\sigma) \le d_1 (\rho, \sigma) \le C \, d_2(\rho, \sigma).
\end{eqnarray}
In particular, there exists a constant that relates any distance to the trace distance $d_{\mathrm{tr}}(\rho,\sigma) = \frac{1}{2} || \rho - \sigma ||_1$.
Therefore, if a correlation quantifier on finite dimensional systems is gd-continuous with respect to the trace distance, it is also gd-continuous with respect to any other distance $d$.
Furthermore, since the trace distance is contractive, Eq.~(\ref{EQ_ABM}) holds for any distance on finite-dimensional systems, at the cost of constants in function $g$.
Accordingly, let us consider the distance induced by the spectral norm $d_{\infty}(\rho, \sigma) = ||\rho - \sigma ||_{\infty}$.
We call the corresponding operator distance $D_{\infty}(\Lambda_1, \Lambda_2)$, and the degree of non-decomposability $\mathrm{ND}_{\infty}(\Lambda)$.
For the connection to the Trotter error, we note the following
\begin{equation}
\mathrm{ND}_{\infty}(U) \le D_{\infty}(U, U_{AM} U_{BM}) \le 2 || U - U_{AM} U_{BM} ||_{\infty},
\label{EQ_ND_INF}
\end{equation}
where the first inequality follows from the fact that $\mathrm{ND}_{\infty}(U)$ is the shortest distance to the set of decomposable maps and $U_{AM} U_{BM}$ is a particular decomposable map.
The second inequality is proven in Appendix~\ref{APP_INFTY}.
Combining the two inequalities, we get $\mathrm{ND}_{\infty}(U) \leq t^2 \norm{ [H_{AM}, H_{BM}] }_{\infty}$.
A concrete example relating the mutual information in a state to the number of Trotter steps is provided in Appendix~\ref{APP_I_TROTTER}.

We have therefore shown a direct link between correlations in the system and the number of Trotter steps one needs to keep the simulation error small.
The amount of violation of Eq.~(\ref{EQ_ABM}) lower bounds the degree of non-decomposability and hence spectral norm of the commutator and accordingly sets the number of required Trotter steps.
Conversely, if it is possible to simulate $U$ with $r$ Trotter steps to precision $\epsilon$, Eq.~(\ref{EQ_R}) shows that commutator norm is bounded 
and consequently Eq.~(\ref{EQ_ND_INF}) implies that correlations $Q_{A:MB}$ admit an upper bound.


\section{Inaccessible mediator}

An interesting opportunity arises where the non-classicality of evolution through mediator could be witnessed without measuring the mediator.
Here we show that this is indeed possible.
We start by introducing the necessary concepts and the related mathematical tools, and then present witnesses of non-decomposable evolution based on measurements on $AB$ only.
Finally, we establish measures of non-decomposability together with their experimentally friendly lower bounds.


\subsection{Marginal maps}

In order to detect non-classicality of interactions solely through the correlations between the coupled objects, we need the notion of `marginals' of decomposable maps.
We propose to introduce it via a related concept of dilation.
A dilation of a map $\Lambda: X \to X$ is an ancillary state $\sigma_R$ and a map $\tilde \Lambda: XR \to XR$ acting on the system and ancilla, such that
\begin{equation}
\Lambda(\rho) = \Tr_R(\tilde \Lambda(\rho \otimes \sigma_R))
\end{equation}
for all $\rho$.
Accordingly, our aim is to exclude the existence of a decomposable dilation of dynamics that are observed on systems $AB$.
In principle, the existence of dilations may depend on the dimension of the Hilbert space of the mediator
which motivates us to introduce \emph{decomposable $m$-dilation} as follows.
A map $\Lambda: AB \to AB$ has a decomposable $m$-dilation if there exists
$\tilde \Lambda: ABM \to ABM$ such that $\tilde \Lambda$ is decomposable and dimension of the mediator satisfies $d_M \le m$.
We denote the set of all maps with a decomposable $m$-dilation as $\overline{\texttt{DEC}}(m)$.

With these definitions we can state our goal precisely:
we wish to infer whether a map on $AB$ admits any decomposable $m$-dilation,
and we wish to do this via measurements of correlations only.
If there does not exist any decomposable dilation, we conclude that the interaction generating the map is non-classical.


\subsection{Detecting non-decomposability}

It turns out that one can obtain an interesting condition that witnesses non-decomposability as a simple corollary to Eq.~(\ref{EQ_ABM}).
In Appendix~\ref{APP_INACCESSIBLE} we prove that any gd-continuous correlation measure $Q$ admits the following bound under the evolution generated by $\lambda \in \overline{\texttt{DEC}}(m)$:
\begin{eqnarray}
Q_{A:B} (\rho_{t}) \le \sup_{\sigma_{XM}} Q_{X:M} (\sigma_{XM}) + I_{A:B}(\rho_{0}),
\label{EQ_AB_DET}
\end{eqnarray}
where $\rho_t = \lambda(\rho_0)$ and we emphasise that $\lambda \in \overline{\texttt{DEC}}(m)$ acts on $AB$ only.
The supremum on the right-hand side runs over all $AM$ or $BM$ states with $d_M \le m$,
and $I_{A:B}(\rho_{AB}) = \inf_{\sigma_A \otimes \sigma_B} g(d(\rho_{AB}, \sigma_A \otimes \sigma_B))$ measures the total correlations across $A:B$.
Note that if the correlation measure that we use is not gd-continuous, we can still obtain a witness of non-decomposability assuming that we start with a product state.
For example, this could be ensured without having access to $M$ by preparing the $AB$ systems in a pure product state.

As an example of using this criterion, note that maximally entangling maps we have discussed before cannot have any decomposable $m$-dilation for $m < \min(d_A, d_B)$.
A question emerges whether there exist evolutions that do not admit decomposable $m$-dilation even when the dimension of the mediator is unbounded.
This is indeed the case.
We show in Appendix~\ref{APP_SWAP} that a $\texttt{SWAP}$ operation on two objects (even two qubits) has no decomposable $m$-dilation for any $m$.
This leads to the conclusion that classical interactions cannot produce a $\texttt{SWAP}$.
The intuitive reason behind this statement is that it takes at least \emph{two} steps to implement swapping with $d_A = d_B = d_M$.
We first exchange $A$ and $M$, then we exchange $B$ and $M$, and we still must exchange $A$ and $M$ again to complete the implementation.
In fact, any $AB$ interaction can be implemented in two steps by first exchanging $A$ and $M$, applying the interaction on $BM$ and finally swapping $A$ and $M$ back.
The conclusion becomes less unexpected once we realise that $\texttt{SWAP}$ is a highly entangling operation.
For example, Alice and Bob can entangle their labs by starting with each having local Bell pairs $\ket{\psi^-_{AA'}} \otimes \ket{\psi^-_{BB'}}$ and swapping the $AB$ subsystems.

We wish to give one more insight into the structure of maps with decomposable dilations.
Clearly, the sets are nested: $\overline{\texttt{DEC}}(m) \subseteq \overline{\texttt{DEC}}(m+1)$.
In fact, the inclusions are strict as we show in Appendix~\ref{APP_STRICT}.


\subsection{Measuring non-decomposability}

In the spirit of the previous section, we would like to extend Eq.~(\ref{EQ_VIOL_AMB}) to bound the distance to $\overline{\texttt{DEC}}(m)$ based solely on correlations measured on systems $AB$.
Of course the $ABM$ operator distance to $\texttt{DEC}$ and the $AB$ operator distance to $\overline{\texttt{DEC}}(m)$ are closely related.
For contractive distances $d$ on states, we have $D(\Lambda_{ABM}, \lambda_{ABM}) \geq D (\Lambda_{AB}, \lambda_{AB})$,
which unfortunately is opposite to what we need.
To overcome this we use so-called completely bounded variant of the operator distance~\cite{Paulsen2003}:
\begin{equation}
\mathcal{D}(\Lambda_1, \Lambda_2) = \sup_{\sigma_{XY}} d( (\Lambda_1 \otimes \openone_Y)(\sigma), (\Lambda_2 \otimes \openone_Y)(\sigma) ),
\end{equation}
where $\Lambda_1, \Lambda_2: X \to X$ and $Y$ is a finite dimensional system.
The benefit of the completely bounded operator distance is that it behaves nicely on dilations.
This makes it easier to jump from the distance to $\texttt{DEC}$ to the distance to $\overline{\texttt{DEC}}(m)$.
Indeed, the completely bounded distance can be written in terms of the dilations as follows:
\begin{equation}
\mathcal{D}(\Lambda_1, \Lambda_2) = \inf_{\tilde \Lambda_i} \mathcal{D}(\tilde \Lambda_1, \tilde \Lambda_2).
\end{equation}
On one hand, for contractive distances on states, the left-hand side cannot be larger than the right-hand side.
On the other hand, the bound can be achieved by an exemplary dilation $\tilde \Lambda_i = \Lambda_i \otimes \openone$.

As a measure of non-decomposability that we will link to the violation of Eq.~(\ref{EQ_AB_DET}),
we propose the analog of the degree of non-decomposability written in terms of completely bounded distance
\begin{eqnarray}
\mathrm{NDm}(\Lambda_{AB}) = \inf_{\lambda_{AB} \in \overline{\texttt{DEC}}(m)} \mathcal{D}(\Lambda_{AB}, \lambda_{AB}).
\end{eqnarray}
With these concepts and tools, it is proven in Appendix~\ref{APP_INACCESSIBLE} that the amount of violation of Eq.~(\ref{EQ_AB_DET}) lower bounds the quantity just introduced:
\begin{eqnarray}
  \mathrm{NDm}(\Lambda_{AB}) &\ge&  g^{-1} (Q_{A:B} (\rho_t) - \mathcal{B}(\rho_{0})),
\end{eqnarray}
where $\mathcal{B}$ is the right-hand side of Eq.~(\ref{EQ_AB_DET}).
Note that all these quantities involve states and maps on $AB$ only.


\subsection{Non-classical gravity}

Our second application of these methods is in foundations.
A prime example of an inaccessible mediator is a mediating field.
The methods described above allow us to make conclusions about the field from the behaviour of objects coupled through it.
Gravitational interaction is especially interesting from this perspective as there is no direct experimental evidence of its quantum properties today.
As discussed in the introduction, observation of quantum entanglement between gravitationally-coupled masses is a plausible near-future experiment closing this gap~\cite{Aspelmeyer2022}.
In this section, we show that our methods allow a concise derivation of the non-classicality witnesses presented in the literature~\cite{Krisnanda2017,Bose2017,Marletto2017}, 
and lead to new conclusions about the interactions that can be drawn from the observation of considerable gravitational entanglement.

Assume first a completely classical situation where both states and interactions are classical.
Recall that within our framework, this means zero-discord state at all times, $D_{AB|M} = 0$ (with one and the same basis on the mediator at all times), and dynamical maps admitting decomposable dilations.
As correlation measure consider quantum entanglement, measured by relative entropy of entanglement.
Then the amount of entanglement $A:B$ that can be produced via these classical maps is
\begin{equation}
E_{A:B}(\rho_t) \le \sup_{\sigma_{XM}} E_{X:M}(\sigma_{X:M}) + I_{A:B}(\rho_0),
\label{EQ_APP_AB}
\end{equation}
where supremum is over all the states of $AM$ or $BM$ allowed in the theory, here $d_M \le m$ and $D_{AB|M} = 0$.
It is reasonable to assume that the initial state in the laboratory will be close to a product state and we therefore take $I_{A:B}(\rho_0) = 0$.
Furthermore, all states admitting $D_{AB|M} = 0$ are disentangled across $A:M$ and $B:M$ and therefore the supremum is also zero.
We therefore arrive at the conclusion that entanglement $A:B$ cannot grow, and hence observation of any gain implies non-classical states or non-classical interactions or both.

If we assume that the interactions are classical (decomposable) but the state might have non-zero discord, then entanglement still satisfies the bound in Eq.~(\ref{EQ_APP_AB}).
Therefore, observation of non-zero value of $E_{A:B}$ means that the supremum on the right-hand side is at least equal to this observed value,
i.e.\ the mediator must be capable of being entangled to $A$ or $B$, and in fact to $AB$ due to monotonicity, to at least the degree that has been measured.
Note that this is stronger than saying that the mediator needs to be discorded.

Finally, by violating the bound in Eq. (\ref{EQ_APP_AB}), it is possible to demonstrate in the laboratory that unknown interactions are not decomposable.
We stress that it is not sufficient to demonstrate that entanglement grows, we have to demonstrate that the entanglement is above a certain threshold.
This threshold depends on the dimension of the mediator and we therefore ask how high entanglement can be generated by gravity.
The answer depends on the concrete setup via which gravitational interaction is studied.
If we take two nearby harmonically trapped masses initially prepared in squeezed states with squeezing parameters $s_A$ and $s_B$,
it has been shown that the gravitational entanglement in terms of logarithmic negativity can be as large as $E_{A:B}^{\max} = |s_A + s_B| / \ln 2$, which holds for large squeezing~\cite{Krisnanda2020}.
Since in principle $s_i \to \infty$, this already shows that gravity cannot be understood as classical interaction with any finite-dimensional mediator.
More practically, the highest optical squeezing achieved today is $s_{A,B} = 1.73$~\cite{Vahlbruch2016},
and assuming it can be transferred to mechanical systems gives entanglement $E_{A:B}^{\max} \approx 5$ ebits,
which would restrict still possible decomposable dilations to use mediators with dimension $m > 2^5$.
It is rather unlikely that this amount of entanglement will be observed in near future, as the time it takes the discussed system to reach $E_{A:B}^{\max}$ in the absence of dissipation is $t_{\max} = \pi \omega L^3/4Gm$, 
independently of high squeezing, where $L$ is the separation between the masses and $\omega$ is the frequency of the trapping potential~\cite{Krisnanda2020}.
For LIGO-like parameters of masses in the order $m \sim 1$ kg, $\omega \sim 0.1$ Hz and $L \sim 1$ cm, this time is already in the order of hours and dissipation pushes it further to tens of hours.
Yet, a violation of the unit bound, and hence disproval of classical interactions via a two-level system, which would already be interesting, could be achieved within a second~\cite{Bose2017,Krisnanda2020,Rijavec2021}.

Another route would be to use gravity to execute dynamics which by other means is known to be non-decomposable.
For example, we have shown below Eq.~(\ref{EQ_AB_DET}) that maximally entangling maps do not admit decomposable dilations for $d_M \le \min(d_A,d_B)$.
The schemes in Refs.~\cite{Bose2017,Marletto2017} indeed use gravity to implement maximal entanglement, but only between two-level quantum systems encoded in path degree of freedom.
It would therefore be interesting to determine whether gravity could be used to maximally entangle masses in more paths.
Along the same line, we showed that $\texttt{SWAP}$ does not admit any decomposable dilation, even with an infinite-dimensional mediator.
Interestingly, Ref.~\cite{Lami2023} argues that gravity could implement the $\texttt{SWAP}$ gate.
In addition, the time it takes to implement the gate is twice as long as the time it takes to implement the maximally entangling unitary, showing that it is not much more demanding than entanglement-based method.
This provides an alternative witness of quantum properties of gravitational interaction that does not rely on a dimension of mediator.


\section{Conclusions}

We have proposed notions of classicality of mediated interactions (commutativity of Hamiltonians and decomposability of dynamical maps) and introduced their mathematical measures.
Our main results are inequalities in terms of any continuous correlation quantifiers with the property that their violations place lower bounds on the amount of introduced non-classicality.
These quantitative methods are therefore experiment ready and applicable in variety of physical situations due to minimalistic assumptions under which they were derived.
As examples, we showed that accurate simulations of dynamics with high correlations necessarily require a large number of Trotter steps,
and that gravitational interaction cannot be understood with the help of commuting particle-field couplings.


\section{Acknowledgements}

We are grateful to Referees for comments that allowed strengthening the conclusion about the non-classicality of gravitational interaction.
This work is supported by the Polish National Agency for Academic Exchange NAWA Project No. PPN/PPO/2018/1/00007/U/00001
and Xiamen University Malaysia Research Fund (Grant No. XMUMRF/2022-C10/IPHY/0002).
RG is supported by the National Science Centre, Poland, within the QuantERA II Programme (No 2021/03/Y/ST2/00178, acronym ExTRaQT) that has received funding from the European Union’s Horizon 2020 research and innovation programme under Grant Agreement No 101017733.
EP acknowledges funding from QuantERA/2/2020, an ERA-Net cofund in Quantum Technologies, under the project eDICT.
This work is supported by Foundation for Polish Science (FNP), IRAP project ICTQT, contract no.\ 2018/MAB/5, co-financed by EU Smart Growth Operational Programme.
M.P. is supported by European Commission via the Horizon Europe research and innovation programme ASPECTS (Grant Agreement No. 101080167).
\appendix

\section{Classicality and decomposability}
\label{APP_NOTION}

\subsection{Classical states}

For completeness, let us start with elementary relations.
A state is said to be classical (or incoherent) if it is diagonal in a preferred basis $\{\ket{m}\}$.
A multipartite state is called quantum-classical (or admits vanishing discord $D_{AB|M} = 0$) if it can be written as $\rho_{\mathrm{qc}} = \sum_m \rho_{AB|m} \otimes \Pi_m$,
where $\Pi_m = \proj{m}$ is the projector on the preferred basis and the systems are enumerated as in Fig.~\ref{FIG_SETUP}.
In words, the whole tripartite state explores only one basis in the Hilbert space of the mediator.
Let us introduce a measurement map along the preferred basis, $\Pi$, whose action on an arbitrary input state is to produce average post-measurement state: $\Pi(\rho) = \sum_m \Pi_m \rho \Pi_m$.
A state $\rho$ is qc (quantum-classical) if and only if $\rho = \Pi(\rho)$.
Alternatively, the definition of classicality can be phrased in terms of commutation with the basis elements.
\begin{proposition}\label{PROP_COMM_CL}
  Let $\Pi(X) = \sum_m \Pi_m X \Pi_m$ be a projection map, where $\Pi_m \Pi_{m'} = \delta_{mm'} \Pi_m$.
  Then
\begin{equation}
X = \Pi(X) \iff \forall m,\, [X, \Pi_m] = 0.
\end{equation}
\end{proposition}
\begin{proof}
The `if' direction is trivial. For the `only if' direction, consider the following argument:
\begin{eqnarray}
X \Pi_m & = & \Pi_m X, \\
X \Pi_m & = & \Pi_m X \Pi_m, \\
X & = & \sum_m \Pi_m X \Pi_m = \Pi(X),
\end{eqnarray}
where we multiplied the first equation by $\Pi_m$ from the right and used $\Pi_m^2 = \Pi_m$, and then we summed the second equation over $m$ and used the completeness relation $\sum_m \Pi_m = \openone$.
\end{proof}


\subsection{Classical interactions}

The definition of classicality of interactions in terms of commutativity is justified by the following proposition.
It shows that the Hamiltonians preserving classicality of states are invariant under dephasing in the preferred basis.
The commutativity is then a corollary.
\begin{proposition}
  Let $H$ be a time-independent Hamiltonian.
  Then $H$ is classical, i.e.\ $H = \Pi(H)$, if and only if for any classical initial state $\rho_0$, $\rho_t = e^{-itH} \rho_0 e^{itH}$ is also classical.
\end{proposition}
\begin{proof}
The `only if' direction. Let us write the assumption explicitly:
\begin{eqnarray}
e^{-i t H} \rho_0 e^{i t H} & = & \Pi\pqty{e^{-i t H} \rho_0 e^{i t H}}, \\
e^{-i t H} [\rho_0, H] e^{i t H} & = & \Pi\pqty{e^{-i t H} [\rho_0, H] e^{i t H}},
\end{eqnarray}
where the second line is the time derivative of the first one and $\rho_0$ denotes the initial (classical) state.
By evaluating at $t=0$, we find that the commutator is invariant:
\begin{eqnarray}
[\rho_0, H] & = & \Pi\pqty{[\rho_0, H]}.
\end{eqnarray}
In particular, taking $\rho_0 = \Pi_m$ shows that for all the basis states:
\begin{eqnarray}
  [\Pi_m, H] & = & \Pi\pqty{[\Pi_m, H]} = 0,
\end{eqnarray}
where the last equation is simple to verify.
Applying Proposition~\ref{PROP_COMM_CL} proves the claim.

The `if' direction.
From the assumption, the Hamiltonian has the block form $H = \sum_m h_m \otimes \Pi_m$, where $h_m$ acts on all the systems other than the mediator.
In this case, the orthonormality of the preferred basis implies
\begin{equation}
e^{\pm i t H} = \sum_m e^{\pm i t h_m} \otimes \Pi_m.
\end{equation}
Accordingly, the initially classical mediator stays classical at all times,
and the remaining systems evolve conditionally depending on the state of the mediator.
\end{proof}
In the case of tripartite systems that we consider, where $H = H_{AM} + H_{BM}$, this shows classicality is preserved when both $H_{AM}$ and $H_{BM}$ are block diagonal with the same basis on system $M$, i.e.\ they commute.

\subsection{Simple eigenstates}
As another argument to justify our definition of classicality, we show that it constraints the eigenstates of the Hamiltonian to be fully product, at least when the local terms are non-degenerate.

\begin{proposition}
  Let $H_{AM}, H_{BM}$ be non-degenerate Hamiltonians.
  Then $[H_{AM}, H_{BM}] = 0$ implies that $H = H_{AM} + H_{BM}$ can be diagonalized with fully product states.
\end{proposition}
\begin{proof}
  Let us assume that $[H_{AM}, H_{BM}] = 0$.
  Note that when a Hermitian matrix $A$ has non-degenerate spectrum, then all eigenvectors of $A \otimes \identity$ must be of the form $\ket{\psi_A} \otimes \ket{\psi_B}$, where $\ket{\psi_A}$ is an eigenvector of $A$ and $\ket{\psi_B}$ is an arbitrary vector.
  Since $[H_{AM}, H_{BM}] = 0$ implies there is a common eigenbasis between $H_{AM}$ and $H_{BM}$, this means that there exists a common eigenbasis for $H = H_{AM} + H_{BM}$ that is product on $A:MB$ and $AM:B$ at the same time, which proves the claim.
\end{proof}


\subsection{One-way decomposability}
\label{APP_ONEWAY}

The following proposition gives an example of decomposable unitary which nevertheless cannot be generated by classical interactions.

\begin{proposition}\label{prop:decomposability-order}
There are no two-qubit unitaries $V_{AM}, V_{BM}$ such that $U_{AM} U_{BM} = V_{BM} V_{AM}$, where
  \begin{align}
    \label{eq:decomposable-nonclassical}
    U_{AM}
    &= \frac{1}{\sqrt{2}} \left( \identity + i Z_A X_M \right)
      \\
    U_{BM}
    &= \frac{1}{\sqrt{2}} \left( \identity + i Z_B Z_M \right),
  \end{align}
and $Z$ and $X$ stand for Pauli matrices.
\end{proposition}
\begin{proof}
	By contradiction.
	Suppose that there exist unitaries $V_{AM}, V_{BM}$ such that $U_{AM} U_{BM} = V_{BM} V_{AM}$.
	Note that we can write $U_{AM}, U_{BM}$ as
	\begin{align}
          U_{AM}
          =
          &\proj{0}_A \otimes \frac{1}{\sqrt{2}} \left( \identity_M + i X_M \right) \nonumber
          \\
          &+ \proj{1}_A \otimes \frac{1}{\sqrt{2}} \left( \identity_M - i X_M \right),
          \\
          U_{BM}
          =
          &\proj{0}_B \otimes \frac{1}{\sqrt{2}} \left( \identity_M + i Z_M \right) \nonumber
          \\
          &+ \proj{1}_B \otimes \frac{1}{\sqrt{2}} \left( \identity_M - i Z_M \right).
	\end{align}
	Therefore, the product $U_{AM} U_{BM}$ is given by 
	\begin{eqnarray}
		&& \proj{00}_{AB} \otimes \half \left( \identity + i X_M + i Y_M + i Z_M \right)
		\nonumber \\
		& + & \proj{01}_{AB} \otimes \half \left( \identity + i X_M - i Y_M - i Z_M \right)
		\nonumber \\
		&+& \proj{10}_{AB} \otimes \half \left( \identity - i X_M - i Y_M + i Z_M \right)
		\nonumber \\
		&+ & \proj{11}_{AB} \otimes \half \left( \identity - i X_M + i Y_M - i Z_M \right).
		\label{eq:decomposability-order-1}
	\end{eqnarray}
	Observe that we can always write $V_{AM} = \sum_{i,j = 0}^1 \ket{i} \bra{j}_A \otimes V_M^{A,ij}$ for some matrices $V_M^{A,ij}$, and similarly for $V_{BM}$.
	However, because we assumed $V_{BM} V_{AM} = U_{AM} U_{BM}$ and the $AB$ part in Eq.~(\ref{eq:decomposability-order-1}) is expressed solely in terms of projectors, we can express $V_{BM} V_{AM}$ as
	\begin{eqnarray}
	V_{BM} V_{AM} = \sum_{i,j} \proj{ij}_{AB} \otimes V_M^{B,jj} V_M^{A,ii},
	\label{eq:decomposability-order-2}
	\end{eqnarray}
	where each product $V_M^{B,jj} V_M^{A,ii}$ is a unitary on $M$.
	Comparing Eqs.~\eqref{eq:decomposability-order-1} and~\eqref{eq:decomposability-order-2}, we find
	\begin{align}
		V_M^{B,00} V_M^{A,00}
		&=
		\half \left( \identity + i X_M + i Y_M + i Z_M \right)
		\\
		V_M^{B,11} V_M^{A,00}
		&=
		\half \left( \identity + i X_M - i Y_M - i Z_M \right)
		\\
		V_M^{B,00} V_M^{A,11}
		&=
		\half \left( \identity - i X_M - i Y_M + i Z_M \right)
		\\
		V_M^{B,11} V_M^{A,11}
		&=
		\half \left( \identity - i X_M + i Y_M - i Z_M \right)
	\end{align}
	However, this leads to the contradiction
	\begin{align}
		\begin{pmatrix}
			0 & 1 \\
			-1 & 0
		\end{pmatrix}
		&=
		\left( V_M^{B,00} V_M^{A,00} \right)
		{\left( V_M^{B,11} V_M^{A,00} \right)}^\dagger
		\\
		&=
		\left( V_M^{B,00} V_M^{A,11} \right)
		{\left( V_M^{B,11} V_M^{A,11} \right)}^\dagger
		=
		\begin{pmatrix}
			0 & -1 \\
			1 & 0
		\end{pmatrix}, \nonumber
	\end{align}
	which completes the proof.
\end{proof}


\subsection{Classicality and commuting decompositions}
\label{APP_CL_COMM}

Here we show the relation between classicality of an interaction and decomposability of the corresponding unitary.
In particular, we show the equivalence between classicality $[H_{AM},H_{BM}] = 0$ and the existence of a continuous, commuting decomposition $U(t) = U_{AM}(t) U_{BM}(t) = U_{BM}(t) U_{AM}(t)$.

\begin{proposition}\label{PROP_DECOMP}
  A one-parameter continuous group of unitaries $U(t) = e^{-itH}$ has a commuting decomposition $U(t) = U_{BM}(t) U_{AM}(t) = U_{AM}(t) U_{BM}(t)$
  such that the map $t \mapsto \pqty{U_{AM}(t), U_{BM}(t)}$ is continuous
  if and only if there exist Hamiltonians $H_{AM}$ and $H_{BM}$ such that $H = H_{AM} + H_{BM}$ and $[H_{AM}, H_{BM}] = 0$.
\end{proposition}
\begin{proof}
  Using the Baker--Campbell--Haussdorf (BCH) formula~\cite{Hall_2015}, one easily sees that if such $H_{AM}, H_{BM}$ exists, then $U(t) = e^{-i t H_{AM}} e^{-i t H_{BM}} = e^{-i t H_{BM}} e^{-i t H_{AM}}$, showing that the unitary has a continuous commuting decomposition.

  To show the other direction, suppose the unitary $e^{-i t H}$ has a continuous commuting decomposition.
  Now, let us take $t$ small enough such that $\norm{U_{AM}(t) - \identity}_{\infty}, \norm{U_{BM}(t) - \identity}_{\infty} < 1$.
  This ensures that $H_{AM} = i \log{U_{AM}(t)}/t, H_{BM} = i \log{U_{BM}(t)}/t$ can be defined through the power series for matrix logarithm.
  Using the series representation $\log{\left(\identity - X\right)} = - \sum_{n=1}^{\infty} \frac{1}{n} X^n$, we notice that these interaction Hamiltonians must commute
  \begin{align}
    [H_{AM}, H_{BM}]
    &=
      - \frac{1}{t^2} [ \log{U_{AM}}, \log{U_{BM}} ] \nonumber
    \\
    &=
      - \frac{1}{t^2} \left[
      \sum_{n = 1}^\infty \frac{{\left( \identity - U_{AM} \right)}^n}{n},
      \sum_{m = 1}^\infty \frac{{\left( \identity - U_{BM} \right)}^m}{m}
      \right] \nonumber
    \\
    &=
      0.
  \end{align}
  Using the BCH formula, we obtain
  \begin{align}
    e^{- i t H}
    = e^{- it H_{AM}} e^{- it H_{BM}}
    = e^{- it \left( H_{AM} + H_{BM} \right)}.
  \end{align}
  Differentiating the expression above with respect to $t$ and using the identity $\left. \frac{d}{dt} e^{(tA)} \right|_{t=0} = A$ shows that $H = H_{AM} + H_{BM}$, which proves the claim.
\end{proof}


\subsection{Consistency}
\label{APP_CONSISTENCY}

Let us start with recalling the two definitions of decomposability given in the main text.

\begin{definition}[unitary]\label{def:unitary-decomposable}
  Let $U$ be a unitary acting on a tripartite system $\mathcal{H}_A \otimes \mathcal{H}_B \otimes \mathcal{H}_M$.
  $U$ is decomposable if there exist unitaries $U_{AM}, U_{BM}$ such that
  \begin{align}
    U_{ABM} = U_{BM} U_{AM}.
  \end{align}
\end{definition}

\begin{definition}[map]\label{def:map-decomposable}

  Let $\lambda$ be a map acting on a tripartite system $\mathcal{H}_A \otimes \mathcal{H}_B \otimes \mathcal{H}_M$.
  $\lambda$ is decomposable if there exist maps $\lambda_{AM}, \lambda_{BM}$ such that
  \begin{align}
    \lambda (\rho) = \lambda_{BM} \lambda_{AM} (\rho).
  \end{align}
\end{definition}

The following proposition shows that these two definitions are consistent.
\begin{proposition}
  A unitary $U$ is decomposable if and only if the map $\lambda(\rho) = U \rho U^\dagger$ is decomposable.
\end{proposition}
\begin{proof}
  If $U$ is decomposable, choosing $\lambda_{AM} (\rho) = U_{AM} \rho U_{AM}^\dagger$ and $\lambda_{BM} (\rho) = U_{BM} \rho U_{BM}^\dagger$ shows that $\lambda$ is also decomposable.

  To show the other implication, suppose there exists two maps $\lambda_{AM}, \lambda_{BM}$ such that $U \rho U^\dagger = \lambda_{BM} \lambda_{AM} (\rho)$.
  It is enough to show that we can choose the maps $\lambda_{AM}, \lambda_{BM}$ to be unitaries.
  This is indeed possible  by the following argument:
  Since $U \rho U^\dagger = \lambda_{BM} \lambda_{AM} (\rho)$, we see that $\sigma \mapsto U^\dagger \lambda_{BM} (\sigma) U$ is a CPTP-inverse of $\lambda_{AM}$.
  Since the only CPTP-maps that have a CPTP-inverse are unitaries~\cite{Nayak_2007}, we conclude that $\lambda_{AM}$ must be a unitary map.
  The fact that $\lambda_{BM}$ is also unitary follows from $\lambda_{BM} (\rho) = U \lambda_{AM}^\dagger (\rho) U^\dagger$.
\end{proof}

Another question regarding the consistency between the two definitions concerns unitary dilations: is decomposability of a map equivalent to the existence of a decomposable unitary dilation?
This would be desirable since this would imply any decomposable map is generated by some `classical' interaction on a larger system.
Here we show that the implication holds in at least in one direction.
\begin{proposition}
	Let $\lambda$ be a decomposable map.
	Then there exists a Stinespring dilation of $\lambda$
	\begin{align}
		\lambda(\rho_{ABM})
		&=
		\Tr_R U_{ABMR} (\rho_{ABM} \otimes \sigma_R) U_{ABMR}^\dagger,
	\end{align}
	such that $U_{ABMR}$ is decomposable.
\end{proposition}
\begin{proof}
	Since $\lambda$ is decomposable, there exist maps $\lambda_{AM}, \lambda_{BM}$ such that $\lambda = \lambda_{BM} \lambda_{AM}$.
	Let us denote a Stinespring dilation of $\lambda_{AM}$ as
	\begin{align}
		\lambda_{AM} (\rho_{AM}) = \Tr_{R_A} U_{AMR_A} (\rho_{AM} \otimes \sigma_{R_A}) U_{AMR_A}^\dagger,
	\end{align}
	where $R_A$ is the purifying system for $\lambda_{AM}$.
	Similarly, $\lambda_{BM}$ must have a dilation with purifying system $R_B$.
        We prove the claim by identifying $R = R_A R_B$, $U_{ABMR} = U_{BMR_B} U_{AMR_A}$, and $\sigma_R = \sigma_{R_A} \otimes \sigma_{R_B}$.
\end{proof}

\subsection{Out-of-time-ordered correlator}
\label{APP_OTOC}

Finally, we comment on the notion of out-of-time-ordered correlator (OTOC) and its relation to the decomposability.
OTOC is often used to study the spread of correlations in a many-body system~\cite{Swingle_2018,Larkin_1969}.
Given two observables $V, W$ (usually chosen to be commuting at time $t=0$), the OTOC is defined as
\begin{align}
  C(t)
  &=
    - \Tr \pqty{ \rho_{\beta} \pqty{ \bqty{V, W(t)} }^2},
\end{align}
where $\rho_{\beta}$ is the thermal state at inverse temperature $\beta$, and $W(t) = e^{-iHt} W e^{iHt}$.
Intuitively, it measures the effect of time evolution on the commutator between two initially commuting observables.
We show that OTOC witnesses the non-decomposability, providing an alternative to our methods.
In particular, let us choose $V$ as an observable on system $A$, and $W$ on system $B$.
Let us assume that the dynamics is decomposable, i.e.\ for any $t$, there exist $U_{AM}, U_{BM}$ such that $e^{-iHt} = U_{BM} U_{AM}$.
Noticing that $[W, U_{AM}] = 0$, an explicit calculation shows that
\begin{align}
  \bqty{V, W(t)}
  &=
    \bqty{V, U_{BM} W U_{BM}^\dagger}
  \\
  &=
    U_{BM} \bqty{V, W} U_{BM}^\dagger,
\end{align}
which is zero, since $V$ and $W$ act on different subsystems.
Therefore the measurement of a non-zero OTOC can witness the non-decomposability of the dynamics.
It remains to be shown whether such an approach can be extended to quantify the degree of non-decomposability.


\section{Accessible mediator}
\label{APP_ACCESSIBLE}

The following proposition proves the `correlation capacity' bound when the initial state is product $\rho_0 = \rho_{AM} \otimes \rho_B$.
\begin{proposition}
\label{TH_CAPACITY}
Let $\lambda$ be a decomposable map. Any correlation measure satisfies:
\begin{equation}
Q_{A:MB}(\lambda(\rho_{AM} \otimes \rho_B)) \le \sup_{\sigma_{AM}} Q_{A:M}(\sigma_{AM}).
\end{equation}
\end{proposition}
\begin{proof}
By assumption $\lambda = \lambda_{BM} \lambda_{AM}$. The bound follows solely from monotonicity of correlations under local operations:
\begin{eqnarray}
Q_{A:MB}(\lambda(\rho_{AM} \otimes \rho_B)) & = & Q_{A:MB}(\lambda_{BM} \lambda_{AM}(\rho_{AM} \otimes \rho_B)) \nonumber \\
& \le & Q_{A:MB}(\lambda_{AM}(\rho_{AM} \otimes \rho_B)).
\end{eqnarray}
Since $Q$ is monotonic under local operations, it must be invariant under invertible local operations.
In particular, adding or discarding an uncorrelated system does not change the value of $Q$.
In our case, system $B$ is completely uncorrelated and therefore $Q_{A:MB}(\lambda_{AM}(\rho_{AM} \otimes \rho_B)) = Q_{A:M}(\lambda_{AM}(\rho_{AM}))$.
Of course, the last quantity is upper bounded by the supremum over all states.
\end{proof}

For a general initial state, we have the following bound by continuity.
\begin{proposition}
\label{TH_CONTINUOUS}
Let $\lambda$ be a decomposable map and $\rho$ any tripartite quantum state. Any gd-continuous correlation measure satisfies:
\begin{equation}
Q_{A:MB}(\lambda(\rho)) \le \sup_{\sigma_{AM}} Q_{A:M}(\sigma_{AM}) + I_{AM:B}(\rho),
\end{equation}
where $I_{AM:B}(\rho) = \inf_{\sigma_{AM}\otimes \sigma_B} g(d(\rho, \sigma_{AM}\otimes \sigma_B))$ is a measure of total correlations in the state $\rho$ across the partition $AM:B$.
\end{proposition}
\begin{proof}
We bound the difference in correlations between arbitrary state and the product state using gd-continuity
\begin{eqnarray}
&&|Q_{A:MB}(\lambda(\rho)) - Q_{A:MB}(\lambda(\sigma_{AM} \otimes \sigma_B))| \nonumber \\
&& \le g(d(\lambda(\rho), \lambda(\sigma_{AM} \otimes \sigma_B))) \nonumber \\
&& \le g(d(\rho, \sigma_{AM} \otimes \sigma_B)),
\end{eqnarray}
where in the last line we used the fact that $g$ is monotonic and $d$ contractive.
The derived inequality holds for any $\sigma_{AM} \otimes \sigma_B$, in particular for the one achieving the infimum of $I_{AM:B} (\rho)$, leading to
\begin{align}
Q_{A:MB}(\lambda(\rho)) \le Q_{A:MB}(\lambda(\sigma_{AM} \otimes \sigma_B)) + I_{AM:B}(\rho).
\end{align}
In the last step we use Proposition~\ref{TH_CAPACITY} to bound the first term on the right.
\end{proof}

In order to simplify the notation, let us denote the bound on correlations due to decomposable dynamics as $B(\rho) = \sup_{\sigma_{AM}} Q_{A:M}(\sigma_{AM}) + I_{AM:B}(\rho)$, and the state at time $t$ as $\rho_t = \Lambda (\rho_0)$.
\begin{proposition}
\label{TH_QUANT}
The degree of non-decomposability $\mathrm{ND}(\Lambda)$ is lower bounded as follows:
\begin{eqnarray}
\mathrm{ND}(\Lambda) & \ge & g^{-1} ( Q_{A:MB}(\rho_t) - B(\rho_0) )
\end{eqnarray}
\begin{proof}
We will prove the theorem by combining the continuity bounds with the statement of Proposition~\ref{TH_CONTINUOUS}.
Consider a fixed, but arbitrary, decomposable map $\lambda$.
Due to gd-continuity we write
\begin{eqnarray}
&& Q_{A:MB}(\rho_t) - Q_{A:MB} (\lambda(\rho_0)) \nonumber \\
&& \le | Q_{A:MB}(\rho_t) - Q_{A:MB} (\lambda(\rho_0))| \nonumber \\
&& \le g(d(\rho_t,\lambda(\rho_0))) \label{EQ_TH_Q1}
\end{eqnarray}
We rearrange and use the bound in Proposition~\ref{TH_CONTINUOUS}:
\begin{eqnarray}
Q_{A:MB}(\rho_t) & \le & Q_{A:MB} (\lambda(\rho_0)) + g(d(\rho_t,\lambda(\rho_0))) \nonumber \\
& \le & B(\rho_0) + g(d(\rho_t,\lambda(\rho_0))). \label{EQ_TH_Q2}
\end{eqnarray}
The amount of violation is now brought to the left-hand side, and below we use the fact that $g$ is invertible and take supremum over states $\rho_0$ to identify the degree of non-decomposability
\begin{eqnarray}
  Q_{A:MB}(\rho_t) - B(\rho_0)& \le & g(d(\rho_t,\lambda(\rho_0))) \nonumber \\
  g^{-1} ( Q_{A:MB}(\rho_t) - B(\rho_0) ) & \le & d(\rho_t,\lambda(\rho_0)), \nonumber \\
  g^{-1} ( Q_{A:MB}(\rho_t) - B(\rho_0) ) & \le & \mathrm{ND}(\Lambda).
\end{eqnarray}
which proves the claim.
\end{proof}
\end{proposition}


\subsection{Spectral norm}
\label{APP_INFTY}

We link the operator norm of unitary maps with the spectral distance between them.
\begin{lemma}
  Let $U, V$ be unitaries.
  Then $D_{\infty} (U, V) \leq 2 \norm{U - V}_{\infty}$.
\end{lemma}
\begin{proof}
  By simple algebra, we verify
  \begin{align}
    U \rho U^\dagger
    - V \rho V^\dagger
    =&
       \half (U-V) \rho (U+V)^\dagger
    \\
     &+ \half (U+V) \rho (U-V)^\dagger,
  \end{align}
  where $\rho$ is a density matrix.
  Taking the spectral norm on both sides, we get
  \begin{align}
    &\norm{U \rho U^\dagger - V \rho V^\dagger}_{\infty}
    \\
    &=
      \norm{
      \half (U-V) \rho (U+V)^\dagger
      + \half (U+V) \rho (U-V)^\dagger
      }_{\infty}
    \\
    &\leq
      \half \norm{(U-V) \rho (U+V)^\dagger}_{\infty}
      + \half \norm{(U+V) \rho (U-V)^\dagger}_{\infty}
    \\
    &\leq
      \norm{U-V}_{\infty}
      \norm{\rho}_{\infty}
      \norm{U+V}_{\infty}
    \\
    &\leq
      2 \norm{U-V}_{\infty}
  \end{align}
  where we used triangle inequality and submultiplicativity of spectral norm $\norm{AB}_{\infty} \leq \norm{A}_{\infty} \norm{B}_{\infty}$.
  Using the bounds $\norm{\rho}_{\infty} \leq \norm{\rho}_1 = 1$ and $\norm{U+V}_{\infty} \leq \norm{U}_{\infty} + \norm{V}_{\infty} = 2$ on the last inequality finishes the proof.
\end{proof}

\subsection{Correlations and the number of Trotter steps}
\label{APP_I_TROTTER}

As a concrete illustration, let us relate the mutual information in a state to number of Trotter steps needed.
In this case, we can use continuity bounds for von Neumann entropy to conclude that if $\frac{1}{2} \norm{\rho - \sigma}_1 = \epsilon$, then~\cite{Audenaert_2007,Winter_2016}
\begin{align}
  \abs{I_{A:MB} (\rho) - I_{A:MB} (\sigma)}
  \leq 2 \epsilon \log \pqty{d_A d_M d_B - 1} + 3 \eta(\epsilon),
\end{align}
where $\eta(x) = -x \log{x} - \pqty{1-x} \log\pqty{1-x}$ is the binary entropy.
We now bound the first term using $\epsilon \leq \sqrt{\epsilon}$, recalling that $\epsilon$ is small, and the second term using $\eta(\epsilon) \leq \sqrt{\epsilon}$ to arrive at
\begin{align}
  \abs{I_{A:MB} (\rho) - I_{A:MB} (\sigma)}
  \leq 5 \log \pqty{d_A d_M d_B} \sqrt{\epsilon}.
\end{align}
Furthermore, since $\norm{X}_1 \leq \rank{X} \cdot \norm{X}_{\infty}$, we have
\begin{align}
  \abs{I_{A:MB} (\rho) - I_{A:MB} (\sigma)}
  \leq C \sqrt{\norm{\rho - \sigma}_{\infty}},
\end{align}
where $C = 5\sqrt{2} \log \pqty{d_A d_M d_B} \sqrt{d_A d_M d_B}$ is a dimension-dependent constant.
This means that we can choose $g(s) = C \sqrt{s}$ to show that mutual information is $gd$-continuous with respect to the spectral distance, and the inverse is $g^{-1} (s) = \pqty{s/C}^2$ when $s \geq 0$.
Combining this with the discussion in Section~\ref{SEC_TROTTER} and Proposition~\ref{TH_QUANT}, we finally get
\begin{align}
  \pqty{\frac
  { I_{A:MB} \pqty{ e^{-itH} \rho_0 e^{itH} } - B\pqty{\rho_0} }
  {C}
  }^2
  \leq t^2 \norm{ [H_{AM}, H_{BM}] }_{\infty},
\end{align}
when $I_{A:MB} \pqty{ e^{-itH} \rho_0 e^{itH} } \geq B\pqty{\rho_0}$.
This means the number of Trotter steps needed to guarantee an $\epsilon$ error is
\begin{align}
  r \geq O\pqty{
  \frac
  { \pqty{I_{A:MB} \pqty{ e^{-itH} \rho_0 e^{itH} } - B\pqty{\rho_0}}^2 }
  {\epsilon}
  }.
\end{align}
Note that while we used some relaxations to derive this bound, we still obtain non-trivial quantitative statements relating the correlations in the system and the commutator norm.
In particular, while quadratic power in the mutual information is sub-optimal, a linear bound cannot exists due to the tightness of the entropic continuity bounds.


\section{Inaccessible mediator}
\label{APP_INACCESSIBLE}

First, we derive a necessary condition on maps admitting a decomposable $m$-dilation.
\begin{proposition}
\label{PR_AB_DET}
A gd-continuous correlation measure $Q$ admits the following bound under the evolution generated by $\lambda \in \overline{\texttt{DEC}}(m)$:
\begin{eqnarray}
Q_{A:B} (\lambda(\rho_{AB})) \le \sup_{\sigma_{AM}} Q_{A:M} (\sigma_{AM}) + I_{A:B}(\rho_{AB}),
\end{eqnarray}
where the supremum is over all $AM$ states with the dimension $d_M \le m$,
and $I_{A:B}(\rho_{AB}) = \inf_{\sigma_A \otimes \sigma_B} g(d(\rho_{AB}, \sigma_A \otimes \sigma_B))$ measures total correlations across $A:B$.
\end{proposition}
\begin{proof}
Consider the following argument:
\begin{eqnarray}
Q_{A:B}(\lambda(\rho_{AB})) & \le & Q_{A:MB}(\tilde \lambda(\rho_{AB} \otimes \sigma_M)) \nonumber \\
& \le & \sup_{\sigma_{AM}} Q_{A:M} (\sigma_{AM}) + I_{AM:B}(\rho_{AB} \otimes \sigma_M) \nonumber \\
& = & \sup_{\sigma_{AM}} Q_{A:M} (\sigma_{AM}) + I_{A:B}(\rho_{AB}),
\end{eqnarray}
where the first line follows from the monotonicity of $Q$ and the existence of a decomposable $m$-dilation,
the second line restates Proposition~\ref{TH_CONTINUOUS} restricted to $m$-dimensional mediator,
and the last line follows from the fact that tracing out an uncorrelated particle is a reversible process and hence equality.
\end{proof}

Note that we have assumed that any map with a decomposable dilation starts with the joint $ABM$ state of a product form $\rho_{AB} \otimes \sigma_M$.
Although this is a restrictive condition, it has been shown that this is essentially the only consistent choice if we require that the dynamics can start from any $AB$ state and the assignment is linear~\cite{Pechukas_1994}.

Next, we show that the violation of the inequality provides a bound on the degree of non-decomposability.

\begin{proposition}
The degree of non-decomposability satisfies the following lower bound:
\begin{eqnarray}
\mathrm{NDm}(\Lambda_{AB}) & \ge & g^{-1} (Q_{A:B} (\Lambda_{AB} (\rho_{AB})) - \mathcal{B}(\rho_{AB})) \nonumber
\end{eqnarray}
where $\mathcal{B}$ is the two-particle version of the bound $B$:
\begin{equation*}
\mathcal{B}(\rho_{AB}) = \sup_{\sigma_{AM}} Q_{A:M}(\sigma_{AM}) + I_{A:B}(\rho_{AB}),
\end{equation*}
and supremum over $\sigma_{AM}$ assumes the dimension of mediator satisfies $d_M \le m$.
\end{proposition}
\begin{proof}
Consider a fixed but arbitrary dilation $\tilde \Lambda$ of the map $\Lambda_{AB}$ and a decomposable map $\tilde \lambda$ (acting on all subsystems) that is a dilation of the map $\lambda_{AB} \in \overline{\texttt{DEC}}(m)$.
The same steps as in Proposition~\ref{TH_QUANT}, Eqs. (\ref{EQ_TH_Q1}) and (\ref{EQ_TH_Q2}), lead to the following inequality
\begin{eqnarray}
&& g^{-1} (Q_{A:B}(\Lambda(\rho_{AB})) - \mathcal{B}(\rho_{AB})) ) \nonumber \\
&\le & d(\tilde \Lambda(\rho_{AB} \otimes \sigma_M), \tilde \lambda(\rho_{AB} \otimes \sigma_M))
\end{eqnarray}
where we have used monotonicity and the definition of dilation to write $Q_{A:B}(\Lambda(\rho_{AB})) \le Q_{A:MB}(\tilde \Lambda(\rho_{AB} \otimes \sigma_M))$
and invariance of total correlations under tracing out uncorrelated system in the bound $B$, which therefore becames $\mathcal{B}$.
The left-hand side is accordingly fully expressed in terms of bipartite quantities and we now similarly bound the right-hand side.

To show the claim, it is enough to show that the distance on the right-hand side gives a lower bound to the degree of non-decomposability.
By taking the supremum over $\rho_{AB}$, the right hand side is upper bounded by the operator distance:
\begin{eqnarray}
\sup_{\rho_{AB}} d(\tilde \Lambda(\rho_{AB} \otimes \sigma_M), \tilde \lambda(\rho_{AB} \otimes \sigma_M)) \le D(\tilde \Lambda, \tilde \lambda),
\end{eqnarray}
where the inequality is due to the optimisation over states of $AB$ only, not over all three systems.
Analogous reasons show that the operator distance is upper bounded by the completely bounded distance
\begin{eqnarray}
D(\tilde \Lambda, \tilde \lambda) \le \mathcal{D} (\tilde \Lambda, \tilde \lambda).
\end{eqnarray}
This time because the right-hand side involves additional optimisation over the ancillary states.
Finally, note that this reasoning holds for any dilation and the best bound is obtained by taking the dilations producing the infimum:
$\inf_{\lambda_{AB} \in \overline{\texttt{DEC}}(m)} \inf_{\tilde \Lambda, \tilde \lambda} \mathcal{D} (\tilde \Lambda, \tilde \lambda) = \inf_{\lambda_{AB} \in \overline{\texttt{DEC}}(m)} \mathcal{D}(\Lambda_{AB}, \lambda_{AB})$.
\end{proof}

With these tools, we investigate in the next sections the structure of maps that admit decomposable $m$-dilations.


\subsection{Non-decomposability of swapping}
\label{APP_SWAP}

\begin{proposition}\label{prop:swap-nondecomposable}
	The map $\texttt{SWAP}$ on two qubits has no decomposable $m$-dilation, for any $m$.
\end{proposition}
\begin{proof}
	We will prove this by contradiction.
	Suppose that $\texttt{SWAP}$ has a decomposable $m$-dilation.
        Let us compare the action of $\texttt{SWAP}$ on $\ket{00}_{AB}$ and on $\ket{01}_{AB}$.
	By definition, there exists two maps $\lambda_{AM}, \lambda_{BM}$ and some initial state $\sigma_M$ such that
	\begin{align}
	  \proj{00}_{AB}
	  &= \texttt{SWAP} (\proj{00}_{AB}) \nonumber
          \\
          &= \Tr_{M} \lambda_{BM} \lambda_{AM} \left( \proj{00}_{AB} \otimes \sigma_M \right),\label{eq:swap-contradiction-1}
	  \\
	  \proj{10}_{AB}
	  &= \texttt{SWAP} (\proj{01}_{AB}) \nonumber
          \\
          &= \Tr_{M} \lambda_{BM} \lambda_{AM} \left( \proj{01}_{AB} \otimes \sigma_M \right).\label{eq:swap-contradiction-2}
	\end{align}
	Let us define $\sigma^0_{AM} = \lambda_{AM} \left( \proj{0}_A \otimes \sigma_M \right)$.
	By Eqs.~\eqref{eq:swap-contradiction-1} and~\eqref{eq:swap-contradiction-2}, we have
	\begin{align}
          \proj{0}_A
          &= \Tr_{B} \texttt{SWAP} (\proj{00}_{AB}) \nonumber
          \\
          &= \Tr_{BM} \lambda_{BM} \left( \proj{0}_B \otimes \sigma^0_{AM} \right),
            \label{eq:swap-contradiction-3}
          \\
          \proj{1}_A
          &= \Tr_{B} \texttt{SWAP} (\proj{01}_{AB}) \nonumber
          \\
          &= \Tr_{BM} \lambda_{BM} \left( \proj{1}_B \otimes \sigma^0_{AM} \right).
            \label{eq:swap-contradiction-4}
	\end{align}
	But because $\lambda_{BM}$ is trace preserving and $\Tr_B$ factors out when applied to product states, we have
	\begin{align}
	  &\Tr_{BM} \lambda_{BM} \left( \proj{0}_B \otimes \sigma^0_{AM} \right)
          \nonumber \\
	  &=
		\Tr_{BM} \left( \proj{0}_B \otimes \sigma^0_{AM} \right)
	  \nonumber \\
	  &=
		\Tr_{M} \sigma^0_{AM}
	  \nonumber \\
	  &=
		\Tr_{BM} \lambda_{BM} \left( \proj{1}_B \otimes \sigma^0_{AM} \right).
	\end{align}
	Combining this with Eqs.~\eqref{eq:swap-contradiction-3} and~\eqref{eq:swap-contradiction-4}, we obtain
	\begin{align}
	  \proj{0}_A
	  &=
		\Tr_{BM} \lambda_{BM} \left( \proj{0}_B \otimes \sigma^0_{AM} \right)
		\\
	  &=
		\Tr_{BM} \lambda_{BM} \left( \proj{1}_B \otimes \sigma^0_{AM} \right)
	  \\
	  &=
		\proj{1}_A,
	\end{align}
	which is clearly a contradiction.
\end{proof}


\subsection{Strict inclusions}
\label{APP_STRICT}

\begin{proposition}
  The inclusion $\overline{\texttt{DEC}}(m) \subsetneq \overline{\texttt{DEC}}(m+1)$ is strict for all $m$.
\end{proposition}
\begin{proof}
  Let us fix $m$ and take $d_A = d_B > d_M = m$.
  Let $\lambda_m (\rho_{AB}) = \Tr_M \texttt{SWAP}_{BM} \lambda_{AM} (\rho_{AB} \otimes \proj{0}_M)$, where $\lambda_{AM}$ is a maximally entangling map.
  By this construction, $\lambda_m$ has a decomposable $m$-dilation, i.e.\ $\lambda_m \in \overline{\texttt{DEC}}(m)$.
  Choosing $\rho_{AB} = \proj{00}_{AB}$ and $Q$ to be relative entropy of entanglement, we obtain $E_{A:B} (\lambda_m(\rho_{AB})) = \log{m}$,
  whereas by Proposition~\ref{PR_AB_DET}, for all maps $\lambda \in \overline{\texttt{DEC}}(m-1)$ we have (recall that $\rho_{AB}$ is product)
  \begin{align}
    E_{A:B} (\lambda (\rho_{AB}))
    &\leq \sup_{\sigma_{AM}} E_{A:M} (\sigma_{AM}) + I_{A:B} (\rho_{AB})
    \\
    &= \log{(m-1)}.
  \end{align}
  Therefore $\lambda_m \not\in \overline{\texttt{DEC}}(m-1)$, and the claim is shown.
\end{proof}

\bibliography{refs.bib}

\end{document}